\documentclass[runningheads,a4paper]{llncs}

\usepackage{amssymb}
\usepackage{graphicx}
\usepackage{thm-restate}

\usepackage{url}   
\newcommand{\keywords}[1]{\par\addvspace\baselineskip
\noindent\keywordname\enspace\ignorespaces#1}

\usepackage{algorithm}
\usepackage[noend]{algorithmic}
\usepackage{microtype}
\usepackage{subfig}
\usepackage{amsmath}
\usepackage{enumerate}
\usepackage{wrapfig}

\graphicspath{{./figures/}}

\newcommand{\frechet}{Fr\'echet}
\newcommand{\df}{d_F}
\newcommand{\dfd}{d_{dF}}

\newcommand{\dfg}{d_{FG}}
\newcommand{\dfr}{d_{FR}}

\date{}

\graphicspath{{./figures/}}

\newcommand{\old}[1]{{{}}}
\def\marrow{\marginpar[\hfill$\longrightarrow$]{$\longleftarrow$}}
\def\todo#1{\textsc{(TODO: \marrow\textsf{#1})}}

\begin{document}

\mainmatter  

\title{The Discrete \frechet\ Gap
	}

\titlerunning{The Discrete \frechet\ Gap}

\author{Omrit Filtser
	\and 	Matthew J. Katz}
\authorrunning{The Discrete \frechet\ Gap}

\institute{
	Ben-Gurion University of the Negev\\
	Beer-Sheva 84105, Israel\\
	\email{\{omritna,matya\}@cs.bgu.ac.il}
}

\maketitle
\begin{abstract}
We introduce the \emph{discrete \frechet\ gap} and its variants as an alternative measure of similarity between polygonal curves. We believe that for some applications the new measure (and its variants) may better reflect our intuitive notion of similarity than the discrete \frechet\ distance (and its variants), since the latter measure is indifferent to (matched) pairs of points that are relatively close to each other. Referring to the frogs analogy by which the discrete \frechet\ distance is often described, the discrete \frechet\ gap is the minimum difference between the longest and shortest positions of the leash needed for the frogs to traverse their point sequences.

We present an optimization scheme, which is suitable for any monotone function defined for pairs of distances such as the gap and ratio functions. We apply this scheme to two variants of the discrete \frechet\ gap, namely, the \emph{one-sided discrete \frechet\ gap with shortcuts} and the \emph{weak discrete \frechet\ gap}, to obtain $O(n^2 \log^2 n)$-time  algorithms for computing them.
\end{abstract}

\old{
\begin{abstract}
	The discrete \frechet\ distance is a well-known similarity measure for two curves. Although it is very useful in many applications, it still has a few drawbacks. For example, since the \frechet\ distance is a bottleneck (min-max) measure, it is very sensitive to outliers. We suggest a new variant of the discrete \frechet\ distance, namely, the discrete \frechet\ gap. We believe that in some applications, the \frechet\ gap variants may improve the compliance between the measured results and the resemblance in reality.
\end{abstract}
}

\section{Introduction}

\old{
	The \frechet\ distance is generally described as follows: Consider a
	person and a dog connected by a leash, each walking along a curve
	from its starting point to its end point. Both can control
	their speed but they are not allowed to backtrack. The \frechet\ distance between two curves $A$ and $B$, denoted by $\df(A,B)$,
	is the minimum length of a leash that is sufficient
	for traversing both curves in this manner.
	
	Intuitively, the \emph{discrete \frechet\ distance}
	replaces the curves by two sequences of points $A=(a_{1},...,a_{n})$
	and $B=(b_{1},...,b_{n})$, and replaces the person and the dog by
	two frogs, the $A$-frog and the $B$-frog, initially placed at $a_{1}$
	and $b_{1}$, respectively. At each move, the $A$-frog or the $B$-frog
	(or both) jumps from its current point to the next one. The frogs are
	not allowed to backtrack. We are interested in the minimum length
	of a leash that connects the frogs and allows the $A$-frog and the $B$-frog to reach $a_{n}$ and $b_{n}$, respectively. More formally, for a given length $\delta$
	of the leash, a jump is allowed only if the distances between the
	two frogs before and after the jump are both at most $\delta$; the
	\emph{discrete \frechet\ distance} between A and B, denoted by $\dfd(A,B)$,
	is then the smallest $\delta>0$ for which there exists a sequence
	of jumps that brings the frogs to $a_{n}$ and $b_{n}$, respectively.
	}

Polygonal curves play an important role in many applied areas, such as 3D modeling in computer vision, map matching in GIS, and protein backbone structural alignment and comparison in computational biology.
Given two curves in a metric space, it is a challenging task to compare them in a way that will reflect our intuitive notion of resemblance.
Various similarity measures have been suggested and investigated, each of them has its advantages and disadvantages. The \frechet\ distance is a useful and well studied similarity measure that has been applied in many fields of research and applications.

The \emph{\frechet\ distance} is often described by an analogy of a man and a dog connected by a leash, each walking along a curve from its starting point to its end point. Both the man and the dog can control their speed but they are not allowed to backtrack. The \frechet\ distance between the two curves is the minimum length of a leash that is sufficient for traversing both curves in this manner.

Intuitively, the \emph{discrete \frechet\ distance} replaces the curves by two sequences of points $A=(a_{1},...,a_{n})$ and $B=(b_{1},...,b_{n})$, and replaces the man and dog by two frogs (connected by a leash), the $A$-frog and the $B$-frog, initially placed at $a_{1}$ and $b_{1}$, respectively. At each move, the $A$-frog or the $B$-frog (or both) jumps from its current point to the next one. The frogs are not allowed to backtrack. We are interested in the minimum length of a leash that allows the $A$-frog and the $B$-frog to reach $a_{n}$ and $b_{n}$, respectively.
The discrete \frechet\ distance is considered a good approximation of the continuous distance, and is easier to compute.

Much research has been done on the \frechet\ distance, the majority of which considers only the continuous version. However, sometimes the discrete \frechet\ distance is more appropriate. For example, in computational biology where each vertex of the polygonal curves represents an alpha-carbon atom. Applying the continuous \frechet\ distance in this case will result in mapping of arbitrary points (i.e., interior points on the edges of the curves), which is not meaningful biologically. See, e.g.,~\cite{AvrahamFKKS14} for a collection of references on the \frechet\ distance and its applications.

In many of the application domains using the \frechet\ distance, the curves or the sampled sequences of points are generated by physical sensors, such as GPS devices. These sensors may generate inaccurate measurements, which we refer to as {\em outliers}.  
Several variants of the \frechet\ distance exist for measuring similarity between curves that might be partially erroneous. In particular, variants for handling outliers have been proposed, since the \frechet\ distance is a bottleneck (min-max) measure and is very sensitive to outliers. Among these variants are the \frechet\ distance with shortcuts~\cite{AvrahamFKKS14,BuchinDS14,DriemelH13}, the partial \frechet\ similarity~\cite{BuchinBW09}, and the average and summed \frechet\ distance~\cite{BrakatsoulasPSW05,EfratFV07}.

In the \emph{one-sided discrete \frechet\ distance with shortcuts}, we allow the $A$-frog to jump to any point that comes later in its sequence, rather than just to the next point. The $B$ frog has to visit all the $B$ points in order,
as in the standard discrete \frechet\ distance problem.

We suggest a new variant of the discrete \frechet\ distance --- the \emph{discrete \frechet\ gap}.
Returning to the frogs analogy, in the discrete \frechet\ gap the leash is elastic and its length is determined by the distance between the frogs. When the frogs are at the same location, the length of the leash is zero.
The rules governing the jumps are the same, i.e., traverse all the points in order, no backtracking.
We are interested in the minimum \emph{gap} of the leash, i.e., the minimum difference between the longest and shortest positions of the leash needed for the frogs to jump from their start points to their end points.

While the discrete \frechet\ distance is determined by the (matched) pairs of points that are
very far from each other and is indifferent towards (matched) pairs of points that are very close to each other, the discrete \frechet\ gap measure 
is sensitive to both. In some cases (though not always), this sensitivity results in better reflection of reality; see Figure~\ref{fig:gap} for examples.

\old{
While the discrete \frechet\ distance is only sensitive to pairs of points that are
very far from each other, thus it does not behave well in the presence of outliers. The \frechet\ gap measure
is also sensitive to pairs of points that are very close to each other. In some cases (but not always), this sensitivity results in better reflection of reality; see Figure \ref{fig:gap} for examples.
}

\begin{figure*}[!h]
	\begin{center}
		\subfloat[] {\includegraphics[scale=0.7]{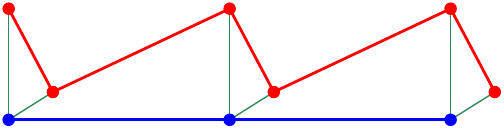}} \hspace{0.4cm}
		\subfloat[]{\includegraphics[scale=0.7]{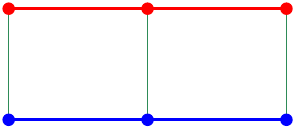}} \hspace{0.4cm}
		\subfloat[]{\includegraphics[scale=0.7]{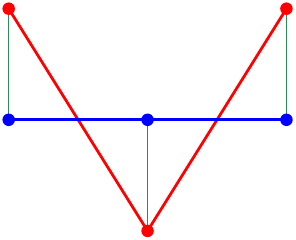}}
	\end{center}
	\vspace{-10pt}
	\caption{\small \frechet\ gap vs. \frechet\ distance: (a) Two non-similar curves, with a large gap. (b) Two similar looking curves. The gap is very small while the \frechet\ distance remains the same as in (a). (c) Two non-similar curves, with a small gap and a large \frechet\ distance.}
	\label{fig:gap}
	\vspace{-10pt}
\end{figure*}

For handling outliers, 
we suggest the \emph{one-sided discrete \frechet\ gap with shortcuts} variant, which we believe has several advantages.
Comparing to the one-sided discrete \frechet\ distance with shortcuts, we believe that the gap variant better reflects the intuitive notion of resemblance between curves in the presence of outliers. Figure~\ref{fig:gap_shortcuts} depicts two curves that look similar, except for a single outlier, with small \frechet\ gap with shortcuts and large \frechet\ distance with shortcuts. Also notice that the gap variant gives a better matching of the points.

\begin{figure*}[!h]
	\centering
	\begin{tabular}{ccc}
		\includegraphics[scale=0.7]{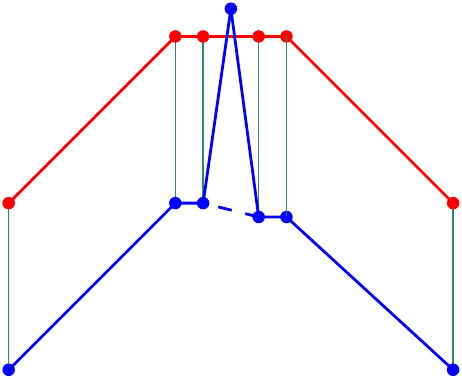} & \hspace{0.5cm}
		\includegraphics[scale=0.7]{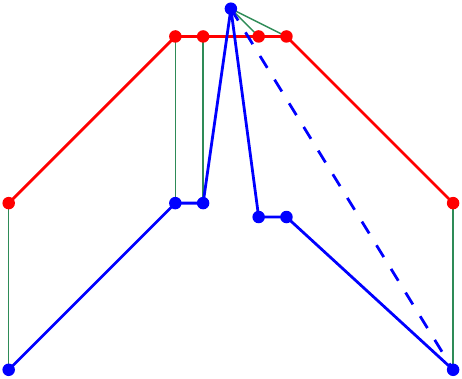} \\
		(a) & (b)
	\end{tabular}
	\centering
	\vspace{-5pt}
	\caption{\small
		(a) The 1-sided \frechet\ gap with shortcuts is small and the outlier is ignored.
		(b) The 1-sided \frechet\ distance with shortcuts is large and the outlier is matched.}
	\label{fig:gap_shortcuts}
\end{figure*}

Avraham et al.~\cite{AvrahamFKKS14} showed that the decision version of the one-sided discrete \frechet\ distance with shortcuts can be solved in linear time, using a greedy algorithm. This algorithm can also be used for solving the decision version of the gap variant. In this paper, we present an efficient optimization algorithm for computing the one-sided discrete \frechet\ gap with shortcuts, which exploits the greediness of the decision algorithm.

Other variants of the discrete \frechet\ distance have corresponding meaningful gap variants. For example, the \emph{weak discrete \frechet\ distance} in which the frogs are allowed to jump also backwards to the previous point in their sequence. The decision version for this variant can also be solved with a greedy algorithm (in quadratic time), and we show how to apply our optimization scheme to efficiently compute the \emph{weak discrete \frechet\ gap}. In general, our scheme can be applied to any variant of the discrete \frechet\ gap that has an efficient greedy decision algorithm.

Notice that the number of potential gaps is $O(n^4)$, while the number of potential distances is only $O(n^2)$. Nevertheless, our algorithms for computing the the one-sided discrete \frechet\ gap with shortcuts and the weak discrete \frechet\ gap run in $O(n^2 \log^2 n)$ time.

Finally, our scheme can be used for computing the \emph{discrete \frechet\ ratio} (and its variants), in which we are interested in the minimum ratio between the longest and the shortest positions of the leash. More generally, one can replace the gap function with any other function $g$ defined for pairs of distances, provided that it is monotone, i.e., for any four distances $c \le a \le b \le d$, it holds that $g(a,b) \le g(c,d)$.

\section{Preliminaries}
\label{sec:prelim}

Let $A=(a_1,\ldots,a_n)$ and $B=(b_1,\ldots,b_n)$ be two sequences of points.
We define a directed graph $G=G(A\times B, E=E_A\cup E_B\cup E_{AB})$, whose vertices are all the possible positions of the two frogs, and whose edges are all the possible moves between positions:
$E_A= \left\{ \Bigr((a_{i},b_{j}),(a_{i+1},b_{j})\Bigl)\right\}$, $E_B= \left\{ \Bigr((a_{i},b_{j}),(a_{i},b_{j+1})\Bigl)\right\}$, $E_{AB}= \left\{ \Bigr((a_{i},b_{j}),(a_{i+1},b_{j+1})\Bigl)\right\}$.

The set $E_A$ corresponds to moves where only the $A$-frog jumps forward, the set $E_B$ corresponds to moves where only the $B$-frog jumps forward, and the set $E_{AB}$ corresponds to moves where both frogs jump forward.
Notice that any valid sequence of moves of the two frogs (with unlimited leash length) corresponds to a path in $G$ from $(a_1,b_1)$ to $(a_n,b_n)$, and vice versa.

It is likely that not all positions in $A \times B$ are valid; for example, when the leash is short. We thus assume that we are given an indicator function $f: A \times B \rightarrow \{0,1\}$, which determines for each position whether it is valid or not.
Now, we say that a position $(a_i,b_j)$ is a \emph{reachable position} (w.r.t. $f$),
if there exists a path $P$ in $G$ from $(a_1,b_1)$ to $(a_i,b_j)$, consisting of only valid positions, i.e.,
for each position $(a_k,b_l) \in P$, it holds that $f(a_k,b_l)=1$.

For any distance $\delta \ge 0$, the function $f_\delta$ is defined as follows:
\vspace{-8pt}
\[
f_\delta(a_i,b_j)=
\begin{cases}
1, & d(a_i,b_j) \le \delta\\
0, & \mbox{otherwise}
\end{cases}\ ,\vspace{-8pt}
\]
where $d(a_i,b_j)$ denotes the Euclidean distance between $a_i$ and $b_j$.
The \emph{discrete \frechet\ distance} $\dfd(A,B)$ is the smallest $\delta \ge 0$
for which $(a_n,b_n)$ is a reachable position w.r.t. $f_\delta$.

For any range of distances $[s,t]$, $0 \le s \le t$, the function $f_{[s,t]}$ is defined as follows:
\vspace{-5pt}
\[
f_{[s,t]}(a_i,b_j)=
\begin{cases}
1, & s \le d(a_i,b_j) \le t\\
0, & \mbox{otherwise}
\end{cases}\ . \vspace{-2pt}
\]
A range $[s,t]$, $s \le t$, is a \emph{feasible (\frechet) range} if $(a_n,b_n)$ is a reachable position w.r.t. $f_{[s,t]}$.

Let $g$ be a bivariate real function with the following monotonicity property:
for any four non-negative real numbers $c \le a \le b \le d$, it holds that $g(a,b) \le g(c,d)$. Then,  
the \emph{discrete \frechet\ range w.r.t. $g$} is the smallest $\delta=g(s,t)$ for which $[s,t]$ is a feasible range.
Two especially important instances are 
the discrete \frechet\ range w.r.t. $t-s$, which we name the \emph{discrete \frechet\ gap} and denote by $\dfg(A,B)$, and
the discrete \frechet\ range w.r.t. $t/s$, which we name the \emph{discrete \frechet\ ratio} and denote by $\dfr(A,B)$.

In the following sections, we use the difference function $t-s$ as a representative function through which we present our ideas and results. However, these results are valid for any such function $g$ (and in particular for the function $t/s$).   

\vspace{-5pt}
\paragraph{One-sided shortcuts variants.}
Let $f$ be an indicator function.
We say that a position $(a_i,b_j)$ is an \emph{s-reachable position} (w.r.t. $f$), if there exists a path $P$ in $G$ from $(a_1,b_1)$ to $(a_i,b_j)$, such that $f(a_1,b_1)=1$, $f(a_i,b_j)=1$, and for each $b_l$, $1 < l < j$, there exists a position $(a_k,b_l) \in P$ that is valid (i.e., $f(a_k,b_l)=1$).
We call such a path an \emph{s-path}. In general, an s-path consists of both valid and non-valid positions.
Consider the path $P'$ (not in $G$) that is obtained from $P$ by deleting the non-valid positions.
Then $P'$ corresponds to a sequence of moves of the two frogs, where the A-frog is allowed to skip points in $A$, and with a leash satisfying $f$. Since in any path in $G$ the two indices (of the A-points and of the B-points) are monotonically non-decreasing, it follows that in $P'$ the B-frog visits each of the points $b_1, \ldots, b_j$, in order, while the A-frog visits only a subset of the points $a_1,\ldots,a_i$, in order.

The \emph{discrete \frechet\ distance with one-sided shortcuts}
$\dfd^S(A,B)$ is the smallest $\delta \ge 0$
for which $(a_n,b_n)$ is an s-reachable position w.r.t. $f_\delta$.

Similarly, a range $[s,t]$ is a \emph{feasible (\frechet) range with one-sided shortcuts} if $(a_n,b_n)$ is an s-reachable position w.r.t. $f_{[s,t]}$, and
the \emph{discrete \frechet\ gap with one-sided shortcuts} $\dfg^S(A,B)$ is the smallest $\delta=t-s \ge 0$ for which $[s,t]$ is a feasible range with one-sided shortcuts.

\vspace{-5pt}
\paragraph{Weak variants.}
Let $G'=G(A\times B, E')$ be the graph obtained from $G$ by adding all backward edges to $E$, i.e., $E'= E \cup \{(v,u)|(u,v)\in E\}$. We say that a position $(a_i,b_j)$ is a \emph{w-reachable position} (w.r.t. $f$), if there exists a path $P$ in $G'$ from $(a_1,b_1)$ to $(a_i,b_j)$ consisting of only valid positions. Such a path corresponds to a sequence of moves of the two frogs, with a leash satisfying $f$, and when backtracking is allowed.

The \emph{weak discrete \frechet\ distance} $\dfd^w(A,B)$ is the smallest $\delta \ge 0$
for which $(a_n,b_n)$ is a w-reachable position w.r.t. $f_\delta$.

Similarly, a range $[s,t]$ is a \emph{feasible weak (\frechet) range} if $(a_n,b_n)$ is a w-reachable position w.r.t. $f_{[s,t]}$, and the \emph{weak discrete \frechet\ gap} $\dfg^w(A,B)$ is the smallest $\delta=t-s \ge 0$ for which $[s,t]$ is a feasible weak range.

\section{Computing the discrete \frechet\ gap} \label{sec:DFG}
Observe that for any feasible (\frechet) range $[s,t]$, it holds that $t \ge \dfd(A,B)$ and $s \le\min\{d(a_1,b_1),d(a_n,b_n)\}$.
Moreover, since we are interested in the minimum feasible range, we may restrict our attention to ranges whose limits are distances between points of $A$ and points of $B$. (Otherwise, we can increase the lower limit and decrease the upper limit until they become such ranges.) Thus, we can search for the minimum feasible range using the following sorted array of distances:
\resizebox{\textwidth}{!}{
$
D=(d^{\min}_m,\dots,d^{\min}_1=\min\{d(a_1,b_1),d(a_n,b_n)\},\dots,\dfd(A,B)=d^{\max}_1,\dots,d^{\max}_k)
$}
\vspace{0.1cm}

\old{
\begin{figure}[h]
	\begin{center}
		\includegraphics[scale=0.8]{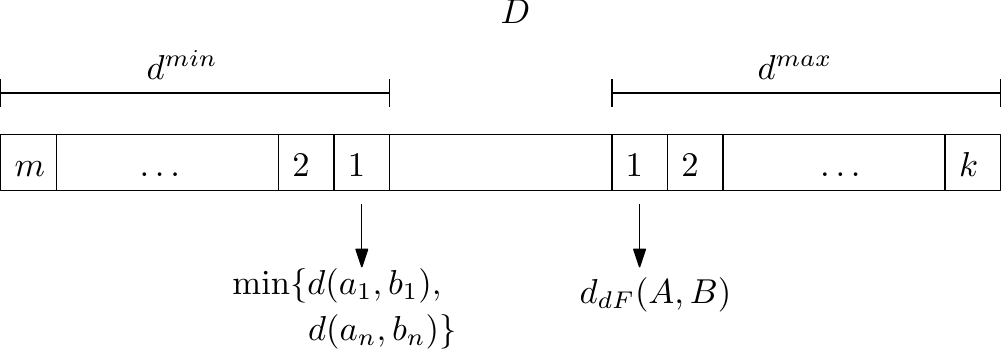}
	\end{center}
	\vspace{-10pt}
	\caption{\small The distances array $D$.}
	\label{fig:distances}
	\vspace{-10pt}
\end{figure}
}

One can compute the discrete \frechet\ gap $\dfg(A,B)$ in $O(n^4)$ time and $O(n^2)$ space, by using the standard $O(n^2)$ decision procedure. Start with the range $[s,t]=[d^{\min}_m,d^{\max}_1]$, and as long as the current range is not a feasible range increase $t$ (by moving to the next distance in $D$). Now, increase $s$ (by moving to the next distance in $D$) and, again, as long as the current range is not a feasible range increase $t$. Repeat the last step (i.e., increase $s$ by a single move and $t$ by a sequence of moves) as much as needed. Finally, return the minimum feasible range that was found during this process.   


A recent result of Ben-Avraham, Kaplan and Sharir~\cite{AvrahamKS15} enables us to reduce the running time to $O(n^3)$. 
Given sequences $A$ and $B$ and an indicator function $f$, they construct a dynamic data structure in $O(n^2)$ time (which also stores the information whether $(a_n,b_n)$ is a reachable position). Following a single change in $f$ (i.e., some valid position becomes non-valid or vice versa), the data structure can be updated in $O(n)$ time. 
Thus, after increasing $s$ or $t$ by moving to the next distance in $D$, we can determine in $O(n)$ time (instead of $O(n^2)$ time) whether $[s,t]$ is a feasible range, since by increasing $s$ a single position becomes non-valid and by increasing $t$ a single position becomes valid.

\section{Computing the discrete \frechet\ gap with one-sided shortcuts} \label{sec:DFGS}

We present an algorithm for computing the discrete \frechet\ gap with one-sided shortcuts in $O(n^2 \log^2n)$ time. Due to space limitations, the proofs of Lemmas~1-3 have been moved to Appendix~\ref{apx:missing_proofs}.

\vspace{-7pt}

\subsection{The decision procedure}

Let $[s,t]$, $0 \le s \le t$ be a range.
We would like to determine whether $[s,t]$ is a feasible (\frechet) range with one-sided shortcuts, i.e., whether $(a_n,b_n)$ is an s-reachable position w.r.t. $f_{[s,t]}$.
In this section, we present a linear-time algorithm for doing so, i.e., for solving the decision version.
Our algorithm is actually much more general and works for any indicator function $f$.
It is similar to the algorithm of Ben-Avraham et al.~\cite{AvrahamFKKS14} for the decision version of the discrete \frechet\ distance with one-sided shortcuts.

Let $f$ be an indicator function, such that $f(a_i,b_j)$ can be evaluated in constant time, for any position $(a_i,b_j)$.
Algorithm~\ref{alg:FD1S} 
computes an $s$-path in $G$ from $(a_1,b_1)$ to $(a_n,b_n)$, if such a path exists. In particular, it determines
whether $(a_n,b_n)$ is an s-reachable position w.r.t. $f$. 
Informally, the B-frog jumps forward (one point at a time) as long as possible, while the A-frog stays in place, then the A-frog makes the smallest forward jump needed to allow the B-frog to continue. The frogs continue advancing in this way, until they either reach $(a_n,b_n)$ or get stuck.

The running time of Algorithm~\ref{alg:FD1S} is clearly $O(n)$, since the number of iterations of the while loop is at most $2n-1$. Notice that we do not construct the graph $G=G(A \times B, E=E_A \cup E_B \cup E_{AB})$, but the path $P$ produced by the algorithm is a path in $G$. Actually, $P$ is a path in $G(A \times B, E=E_A \cup E_B)$, since the algorithm does not advance the frogs simultaneously. We will use this observation later. The correctness of the algorithm is given by the following lemma.

\vspace{-3pt}

\begin{restatable}{lemma}{dFSDecision}
\label{lem:dF1S-decision}
Given two sequences of points $A=(a_{1},\ldots,a_{n})$ and $B=(b_{1},\ldots,b_{n})$ and an indicator function $f$, dF1S-decision($A,B,f$) returns ``yes'' iff $(a_{n},b_{n})$ is an s-reachable position in $G=G(A \times B, E=E_A \cup E_B \cup E_{AB})$ w.r.t. f.
\end{restatable}

\vspace{-3pt}

Since we did not make any assumptions regarding the function $f$, Algorithm~\ref{alg:FD1S} can be used as the decision procedure for both the discrete \frechet\ distance with one-sided shortcuts and the discrete \frechet\ gap with one-sided shortcuts: given a real number $\delta \ge 0$ or an interval $[s,t]$, we simply replace $f$ by $f_\delta$ or $f_{[s,t]}$, respectively.

\vspace{-10pt}

\begin{algorithm}[h]
	\vspace*{.2cm}
	\begin{enumerate}
		\item If $f(a_1,b_1)=0$ or $f(a_n,b_n)=0$ return ``no''.
		\item $curr \leftarrow(a_1,b_1)$.
		\item While $true$
		\begin{description} 
			\item Assume $curr$ is $(a_i,b_j)$.
			\item $P \leftarrow curr$. 		
			\item[\ \,If] $curr$ is valid ($f(curr)=1$)
			\begin{description} 
				\item[If] $j<n$, set $curr \leftarrow (a_i,b_{j+1})$ \\ (the B-frog jumps to its next point $b_{j+1}$, while the A-frog stays at $a_i$).
				\item[Else] ($j=n$) return ``yes'' \\ (the A-frog can jump directly to $a_n$, while the B-frog is already at $b_n$). 
			\end{description}
			\item[\ \,Else] ($curr$ is non-valid, i.e., $f(curr)=0$)
			\begin{description}
				\item[If] $i<n$, set $curr \leftarrow (a_{i+1},b_j)$  \\
				(the A-frog skips $a_i$, while the B-frog stays at $b_j$).
				\item[Else] ($i=n$) return ``no'' (the frogs cannot reach position $(a_n,b_n)$).
			\end{description}
		\end{description}
	\end{enumerate}
	\vspace*{-.25cm}	
	\caption{dF1S-decision($A,B,f$)}\label{alg:FD1S}
\end{algorithm}

\vspace{-25pt}

\subsection{The search algorithm}
Consider the following sorted distances array: \vspace{-5pt}
\[
D=(d^{\min}_m,\dots,d^{\min}_1=\min\{d(a_1,b_1),d(a_n,b_n)\},\dots,\dfd^S(A,B)=d^{\max}_1,\dots,d^{\max}_k)\ .
\vspace{-5pt}
\]\ 
For any feasible range $[s,t]$ with one-sided shortcuts, it holds that $t \ge \dfd^S(A,B)$ and $s \le \min\{d(a_1,b_1),d(a_n,b_n)\}$. As in Section~\ref{sec:DFG}, we may restrict our attention to ranges whose limits are distances between points of $A$ and points of $B$.

For the rest of this section, whenever we refer to a feasible range, we actually mean a feasible range with one-sided shortcuts.

\begin{wrapfigure}{r}{0.45\textwidth}
	\vspace{-41pt}
	\begin{center}
		\includegraphics[width=0.45\textwidth]{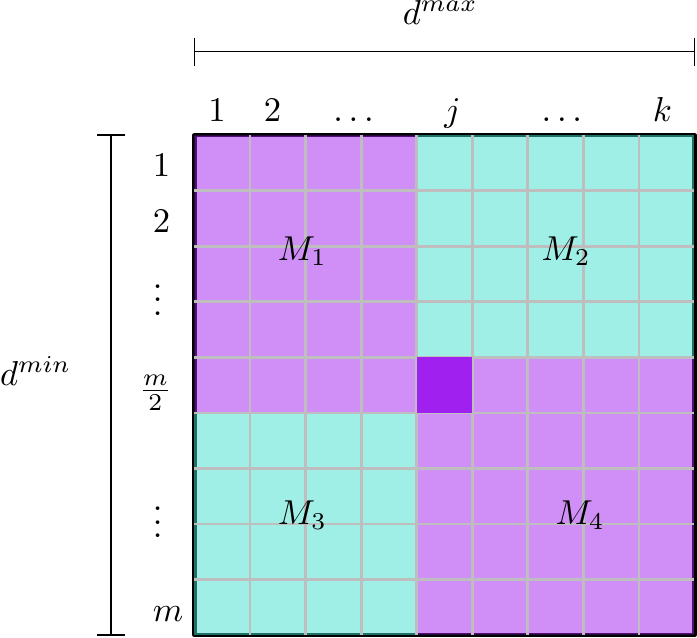}
	\end{center}
	\vspace{-15pt}
	\caption{\small The matrix of possible gaps.}
	\vspace{-15pt}
	\label{fig:GapMatrix}
\end{wrapfigure}
Let $M$ be the matrix whose rows correspond to $d^{\min}_1,\dots,d^{\min}_m$ and whose columns correspond to $d^{\max}_1,\dots,d^{\max}_k$ (see Figure~\ref{fig:GapMatrix}). A cell $M_{i,j}$ of the matrix corresponds to the range $[d^{\min}_i,d^{\max}_j]$. 
$M$ is sorted in the sense that range $M_{i,j}$ contains all the ranges $M_{i',j'}$ with $i' \le i, j' \le j$. Thus, we can perform a binary search in the middle row to find the smallest feasible range $M_{\frac{m}{2},j}=[d^{\min}_\frac{m}{2},d^{\max}_j]$ among the ranges in this row.

$M_{\frac{m}{2},j}$ induces a partition of $M$ into 4 submatrices: $M_1,M_2,M_3,M_4$ (see Figure~\ref{fig:GapMatrix}). Each of the ranges in $M_1$ is contained in a range of the middle row which is not a feasible range, hence none of the ranges in $M_1$ is a feasible range. Each of the ranges in $M_4$ contains $M_{\frac{m}{2},j}$ and hence is at least as large as $M_{\frac{m}{2},j}$. Thus, we may ignore $M_1$ and $M_4$ and focus only on the ranges in the submatrices $M_2$ and $M_3$.

\paragraph{Sketch of the algorithm.}
Our goal is to find the smallest range in $M$ for which the decision algorithm returns ``yes''.
This can be done in $O(n^3 \log n)$ time by first finding in each of $M$'s rows (via binary search) the smallest range for which the decision algorithm returns ``yes'', and then picking the smallest among these $O(n^2)$ ranges.
Below, we sketch a nearly quadratic algorithm for finding the smallest feasible range. 
 
We perform a recursive search in the matrix $M$.
The input to the recursive algorithm is a submatrix $M'$ of $M$ and a graph $G'$ by which one can decide for each range in $M'$ whether it is a feasible range or not. In each recursive call, we perform a binary search in the middle row of $M'$ to find the smallest feasible range in this row, using the graph $G'$. Then, we construct the two graphs for the two submatrices of $M'$ in which we still need to search in the next level of the recursion. 

Notice that we could use the graph $G=G(A \times B, E=E_A \cup E_B)$ in each of the recursive calls, but this would yield an algorithm of running time $O(n^3 \log n)$. Instead, in each recursive call we use a graph whose size is proportional to the number of rows and columns in the submatrix for this call. The introduction of these graphs and their efficient construction is the main contribution of this section. 

We represent $M$ and its submatrices by the indices of the array $D$ that correspond to the rows and columns of $M$. For example, we represent $M$ by $[1,m]\times[1,k]$, $M_1$ by $[1,\frac{m}{2}]\times[1,j-1]$, $M_2$ by $[1,\frac{m}{2}-1]\times[j,k]$, $M_3$ by $[\frac{m}{2}+1,m]\times[1,j-1]$, and $M_4$ by $[\frac{m}{2},m]\times[j,k]$.

The skeleton of the algorithm is given below. Recall that $g$ is a bivariate real function with the property that for any four non-negative real numbers $c \le a \le b \le d$, it holds that $g(a,b) \le g(c,d)$ (see Section~\ref{sec:prelim}). In our case, $g(s,t)=t-s$.

\vspace{-15pt}
\begin{algorithm}[h]
    \vspace*{.25cm}
    \begin{enumerate}
        \item Perform a binary search in the middle row of matrix $M_{[p,p']\times[q,q']}$ to find the smallest feasible range $[d^{\min}_i,d^{\max}_j]$, $i=\frac{p'+p-1}{2}$, using the decision procedure with the graph $G_{[p,p']\times[q,q']}$.
        \item Construct the graphs $G_{[p,i-1]\times[j,q']}$ and $G_{[i+1,p']\times[q,j-1]}$ for the submatrices $M_{[p,i-1]\times[j,q']}$ and $M_{[i+1,p']\times[q,j-1]}$, respectively.
        \item Return \begin{flalign*}
				        \min\{ & g(d^{\max}_j,d^{\min}_i),\\
				        & \mbox{S-Alg}(M_{[p,i-1]\times[j,q']},G_{[p,i-1]\times[j,q']}),\\
				        & \mbox{S-Alg}(M_{[i+1,p']\times[q,j-1]},G_{[i+1,p']\times[q,j-1]})\}.
				        \end{flalign*}
    \end{enumerate}
    \vspace*{-.25cm}
    \caption{S-Alg$\left(M_{[p,p']\times[q,q']},G_{[p,p']\times[q,q']}\right)$}\label{alg:search}
\end{algorithm}
\vspace{-15pt}

The number of potential feasible ranges is equal to the number of cells in $M$, which is $O(n^4)$.
But, since we are looking for the smallest feasible range, we do not need to generate all of them.
We only use $M$ to illustrate the search algorithm, its cells correspond to the potential feasible ranges, but do not contain any values. We thus define the {\em size} of a submatrix of $M$ by the sum of its number of rows and number of columns, for example, $M$ is of size $m+k$, $M_2$ is of size $\frac{m}{2}+k-j$, and $M_3$ is of size $\frac{m}{2}+j-1$.

Notice that the ranges in $M_2$ and $M_3$ consist of all the ranges that intersect $g=[d^{\min}_\frac{m}{2},d^{\max}_j]$ and are neither contained in $g$ nor contain $g$: $M_2$ consists of all the ranges with minimum distance larger than $d^{\min}_\frac{m}{2}$ and maximum distance at least as large as $d^{\max}_j$, and $M_3$ consists of all the ranges with minimum distance smaller than $d^{\min}_\frac{m}{2}$ and maximum distance smaller than $d^{\max}_j$. This implies that for {\em any} range in $M_2$, all the distances in $[d^{mim}_1,d^{\max}_j]$ are in the range and all the distances in $[d^{\min}_m,d^{\min}_\frac{m}{2}]$ are not in the range (see Figure~\ref{fig:FixedValues}). 
\begin{figure}[!h]
	\vspace{-15pt}
	\begin{center}
		\includegraphics[scale=0.8]{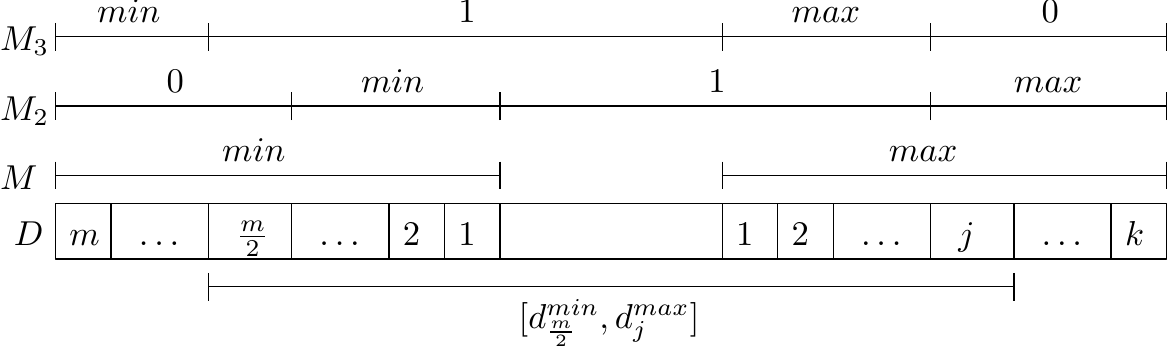}
	\end{center}
	\vspace{-10pt}
	\caption{\small The concept of fixed values.}
	\vspace{-15pt}
	\label{fig:FixedValues}
\end{figure}

More generally, let $M'$ be any of the submatrices associated with the $i$'th level of the recursion tree.
A distance $d \in D$ is {\em fixed} w.r.t. $M'$ if one of the following two statements is correct:
\begin{enumerate}[(i)]
	\item $d$ belongs to all the ranges in $M'$, or
	\item $d$ does not belong to any of the ranges in $M'$.
\end{enumerate}
Otherwise, $d$ is {\em non-fixed} (w.r.t. $M'$). 
The non-fixed distances w.r.t. $M'$ are thus the distances that correspond to the rows and to the columns of $M'$ (see Figure~\ref{fig:FixedValues}).

\begin{restatable}{lemma}{totalSize}
\label{lem:totalSize}
The total size of the matrices in level $i$ of the recursion tree is at most $m+k$, for any level $i$.
\end{restatable}
It follows that the number of non-fixed distances in each level of the recursion is at most $m+k$. We wish to use this fact to reduce the running time of our algorithm.

\subsection{The construction of the graphs}

Let $d_{i,j}$ denote the distance between $a_i$ and $b_j$. We say that distance $d$ {\em belongs} to submatrix $M'$ of $M$ (and write $d \in M'$) if $d$ is one of the distances corresponding to the rows and columns of $M'$.

We first construct the graph $G=G(V=A \times B, E=E_A \cup E_B)$. Each vertex $v_{i,j}=(a_i,b_j) \in V$ has two outgoing edges: 
\begin{enumerate}
\item $jumpB(v_{i,j})=v_{i,j+1}$ (and if $j=n$, $jumpB(v_{i,j})=v_{n,n}$), and 
\item $jumpA(v_{i,j})=v_{i+1,j}$ (and if $i=n$, $jumpA(v_{i,j})=null$). 
\end{enumerate}
Then, we construct the graph $G_0$ for the matrix $M_0=M$ from $G$, by removing the vertices of $G$ whose corresponding distances are fixed w.r.t. $M_0$ (these are all the distances between $d^{\min}_1$ and $d^{\max}_1$ in the array $D$), and updating the edge set as described below.

In general, let $M_{i-1}$ and $G_{i-1}$, $i > 0$, be a matrix in level $i-1$ and the graph constructed for it;
then $V_{i-1}=\{v_{i,j} \ | \ d_{i,j} \in M_{i-1}\} \cup \{v_{1,1},v_{n,n}\}$.
Let $M_i$ be one of the two submatrices of $M_{i-1}$ in level $i$. 
We describe how $G_i$ is obtained from $G_{i-1}$ (and how $G_0$ is obtained from $G$).
Some of the vertices of $G_{i-1}$ are {\em fixed} w.r.t. $G_i$, i.e., their corresponding distances are fixed w.r.t. $M_i$ (i.e.,
they do not belong to $M_i$). We say that such a fixed vertex $v$ (whose corresponding distance is $d$) is {\em valid} (resp., {\em non-valid}), if $d$ belongs to all ranges in $M_i$ (resp., if $d$ does not belong to any of the ranges in $M_i$).
Since backtracking is forbidden, $G_{i-1}$ is acyclic and one can topologically sort its vertices. 
We do so, and then process the vertices, one by one, in reverse order (i.e., from last to first).
More precisely, for each vertex $v \in V_{i-1}$, we run the code fragment below, where $d$ is the distance corresponding to $v$. 
The code fragment sets the pointer $next(v)$, for each fixed vertex $v \in V_{i-1}$, so that $next(v)$ is the first non-fixed vertex in the path (in $G_{i-1}$) beginning at $v$ that is induced by the greedy decision algorithm.

\begin{description}
	\item[set] $next(v_{n,n})\leftarrow v_{n,n}$
    \item[if] $v$ is fixed and $v\neq v_{n,n}$
        \begin{description}
            \item[if] $v$ is valid ($f_{[s,t]}(v)=1$ for all ranges $[s,t]\in M_i$)
                \begin{description}
                    \item[if] $jumpB(v)$ is fixed, $next(v)\leftarrow next(jumpB(v))$
                    \item[if] $jumpB(v)$ is not fixed, $next(v)\leftarrow jumpB(v)$
                \end{description}
            \item[if] $v$ is non-valid ($f_{[s,t]}(v)=0$ for all ranges $[s,t]\in M_i$)
                \begin{description}
                    \item[if] $jumpA(v)$ is fixed, $next(v)\leftarrow next(jumpA(v))$
                    \item[if] $jumpA(v)$ is non-fixed, $next(v)\leftarrow jumpA(v)$
                \end{description}
        \end{description}
    \item[if] $v$ is non-fixed or $v=v_{1,1}$
        \begin{description}
            \item[if] $jumpB(v)$ is fixed, $jumpB(v)\leftarrow next(jumpB(v))$ (else, do nothing)
            \item[if] $jumpA(v)$ is fixed, $jumpA(v)\leftarrow next(jumpA(v))$ (else, do nothing)
        \end{description}
\end{description}

Notice that after processing all the vertices of $G_{i-1}$, it holds that (i) for any fixed vertex $v \ne v_{n,n}$, $next(v)$ is non-fixed, and
(ii) for any non-fixed vertex $v$, $jumpB(v)$ and $jumpA(v)$ are also non-fixed (unless maybe when $jumpB(v)=v_{n,n}$ or $jumpA(v)=v_{n,n}$).
We thus set $V_i = V_{i-1} \setminus \{v \in V_{i-1} \ | \ v \mbox{ is fixed w.r.t. } M_i$ $\mbox{and } v \neq v_{1,1},v_{n,n}\} = \{v_{i,j} \ | \ d_{i,j} \in M_i\} \cup \{v_{1,1},v_{n,n}\}$ and define $G_i$ as the graph induced by $V_i$.
See Figure~\ref{fig:DecisionGraph} for an example.

\begin{figure*}[!h]
	\begin{center}
		\subfloat[]{\includegraphics[scale=0.7]{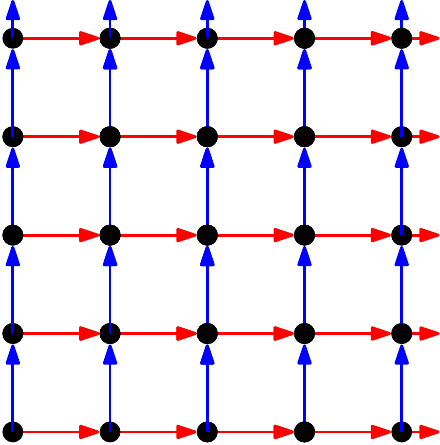}}  \hspace{1cm}
		\subfloat[]{\includegraphics[scale=0.7]{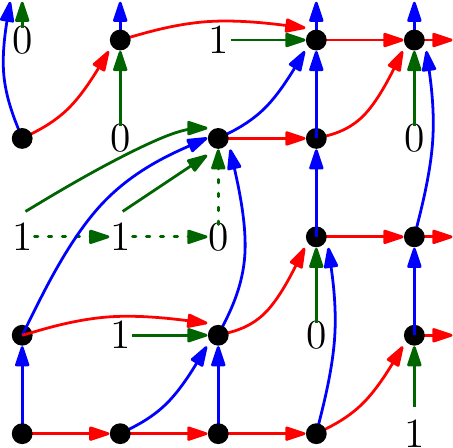}} 
	\end{center}
	\vspace{-10pt}
	\caption{\small (a) The graph $G_0$. (b) A graph that was constructed from $G_0$ in level~1 of the recursion tree. The points represent the vertices of the graphs (all non-fixed). The 1's (resp., the 0's) are the vertices of $G_0$ that are fixed valid (resp., fixed non-valid) w.r.t. to the new graph. The green arrows mark $next()$, the red ones mark $jumpB()$, and the blue ones mark $jumpA()$. The new graph consists only of the points and the red and blue arrows.}
	\label{fig:DecisionGraph}
\end{figure*}

Given a range $[s,t]$ corresponding to the matrix $M_i$, we can apply the decision procedure to the graph $G_i$: If both the distances corresponding to $v_{1,1}$ and to $v_{n,n}$, respectively, are within the range (i.e., $f_{[s,t]}(v_{1,1})=1$ and $f_{[s,t]}(v_{n,n})=1$), then perform the following loop (otherwise, return ``no''). Let $v$ be the current vertex (where initially $v=v_{1,1}$). If $v$ is valid (i.e., if $f_{[s,t]}(v)=1$), go to $jumpB(v)$, else go to $jumpA(v)$. Return ``yes'' if and only if you have reached $v_{n,n}$. This takes only $O(|V_i|)$ time.

\paragraph{Correctness.}
It remains to prove that the decision obtained when applying the decision procedure to $G_i$ is the same as the one obtained when applying it to the original graph $G$.

\begin{restatable}{lemma}{shortcutsCorrectness}
\label{lem:shortcuts-correctness}
Given a range $[s,t]$ corresponding to $M_i$,
the decision algorithm applied to $G_i$ returns ``yes'' if and only if the decision algorithm applied to $G$ returns ``yes''.
\end{restatable}

\paragraph{Running time.}
Consider the recursion tree. It consists of $O(\log n)$ levels, where
the $i$'th level is associated with $2^i$ disjoint submatrices of the matrix $M$.
Level 0 is associated with the matrix $M_0=M$, level 1 is associated with the submatrices $M_2$ and $M_3$ of $M$ (see Figure~\ref{fig:GapMatrix}), etc.

A {\em range test} is a test that determines for two distances $d_1<d_2$ in the sorted array of distances $D$ \old{(see Figure~\ref{fig:distances})} whether the range $[d_1,d_2]$ is a feasible range.
In the $i$'th level we perform $O(\log n)$ range tests in each of the $2^i$ submatrices associated with this level.
We claim that the total time spent on the $i$'th level is $O(n^2\log n)$.
This bound includes the preparations towards the next level.
Therefore, the running time of the entire algorithm is $O(n^2 \log^2 n)$.

We now focus on the analysis of the $i$'th level.
Let $M'$ be any of the submatrices associated with the $i$'th level.
Our algorithm guarantees that the cost of a range test, for a range corresponding to $M'$, is linear in the size of $M'$.
By Lemma~\ref{lem:totalSize}, the total size of the submatrices in level $i$ is $m+k \le n^2$, and therefore
the total cost of all range tests performed in the $i$'th level is $O(n^2 \log n)$.
Finally, the preparations towards the next level require only $O(n^2)$ time.

The following theorem summarizes the main result of this section.
\begin{theorem}
Let $A=(a_1,\ldots,a_n)$ and $B=(b_1,\ldots,b_n)$ be two sequences of points. Then,
the discrete \frechet\ gap with one-sided shortcuts $\dfg^S(A,B)$ and 
the discrete \frechet\ ratio with one-sided shortcuts $\dfr^S(A,B)$ can be computed in $O(n^2 \log^2 n)$ time.
\end{theorem}

\old{
Sometimes, computing the smallest $\delta=t/s$ for which $[s,t]$ is a feasible range with one-sided shortcuts might be more meaningful. We call this variant the {\em discrete \frechet\ ratio with one-sided shortcuts} and denote it by $\dfr^S(A,B)$. In order to compute $\dfr^S(A,B)$, we only need to change Algorithm~\ref{alg:search} so that it returns the minimum value among $d^{\max}_j/d^{\min}_i$ and the values returned by the two recursive calls (instead of returning the minimum value among  $d^{\max}_j-d^{\min}_i$ and the values returned by the two recursive calls). We thus have

\begin{corollary}
The discrete \frechet\ ratio with one-sided shortcuts $\dfr^S(A,B)$ can be computed in $O(n^2 \log^2 n)$ time.
\end{corollary}
}

\section{Computing the weak discrete \frechet\ gap}
\label{sec:DFGW}
We apply the high-level search algorithm to the weak discrete \frechet\ gap variant. For this we need to (i) describe a suitable greedy decision algorithm and (ii) show how to efficiently construct the graphs for the two submatrices of the next level. Due to space limitations, we only briefly discuss (i); full details of both (i) and (ii) (which is the key issue here) are given in Appendix~\ref{apx:DFGW}.

Let $G'=G(A \times B, E')$, where $E'=E_A \cup E_B \cup \{(v,u)|(u,v) \in E_A \cup E_B\}$. That is, $G'$ is obtained from the graph $G$ of the
`strong' version, by adding the backward edges. (For simplicity, we assume in this section that the frogs are not allowed to jump simultaneously.)   
Notice that $G'$ is a planar graph. We view $G'$ as a maze. Each vertex is a room with four doors, one for each outgoing edge, that lead to the adjacent rooms (see Figure~\ref{fig:weak}). A man standing in room $(a_1,b_1)$ wants to reach room $(a_n,b_n)$, but without entering forbidden rooms (i.e., rooms corresponding to non-valid positions).

\begin{wrapfigure}{r}{0.4\textwidth}
	\vspace{-30pt}
	\begin{center}
		\includegraphics[width=0.3\textwidth]{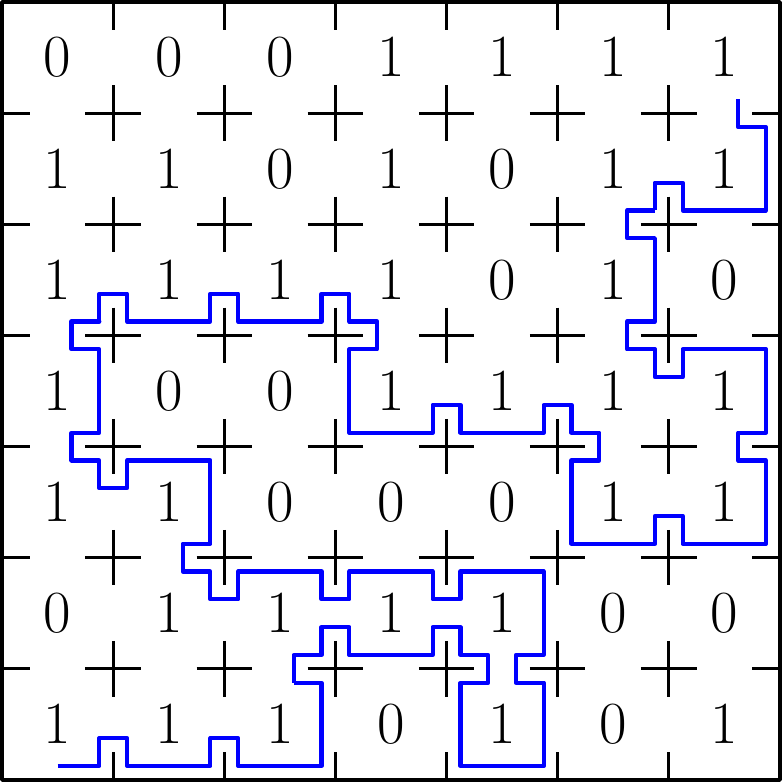}
	\end{center}
	\vspace{-15pt}
	\caption{\small Following the walls.}
	\vspace{-30pt}
	\label{fig:weak}
\end{wrapfigure}

A well known algorithm for traversing a maze is the \emph{wall-follower} rule (also known as the right-hand rule): keep your right hand in contact with a wall of the maze throughout the motion. The algorithm guarantees that you will eventually reach the exit, if possible.
Thus, one can find in $O(n^2)$ time a weak-path in $G'$ from $(a_1,b_1)$ to $(a_n,b_n)$, if such a path exists (or reach $(a_1,b_1)$ if no such path exists).

The following theorem summarizes the main result of this section.
 \begin{theorem}
 	Let $A=(a_1,\ldots,a_n)$ and $B=(b_1,\ldots,b_n)$ be two sequences of points. Then,
 	the weak discrete \frechet\ gap $\dfg^w(A,B)$ and the weak discrete \frechet\ ratio $\dfr^w(A,B)$ can be computed in $O(n^2 \log^2 n)$ time.
 \end{theorem}

\old{	
	\section{Conclusions and open questions}
	\todo{}
	
	\frechet\ under translations?
	
	\frechet\ quotient?
	
	other variants?
	
	It is surprising that the shortcuts and weak variants have more efficient algorithms...
}

\vspace{-15pt}
\bibliography{refs}
\vspace{-10pt}

\newpage
\appendix

\section{Computing the weak discrete \frechet\ gap}
\label{apx:DFGW}
In this section we apply the high-level search algorithm (i.e., Algorithm~\ref{alg:search}) to the problem of computing the weak discrete \frechet\ gap. For this we need to (i) describe a suitable greedy decision algorithm and (ii) show how to efficiently construct the graphs for the two submatrices of the next level.
For simplicity, we assume in this section that the frogs are not allowed to jump simultaneously, however, our solution can be easily adapted to the case where simultaneous jumps are allowed. 

\subsection{The decision procedure}
Let $G'=G(A \times B, E')$, where $E'=E_A \cup E_B \cup \{(v,u)|(u,v) \in E_A \cup E_B\}$. That is, $G'$ is obtained from the graph $G$ of the `strong' version, which contains only the forward edges, by adding the backward edges.
Let $f: A \times B \rightarrow \{0,1\}$ be an indicator function, which determines for each position whether it is valid or not.
We provide an $O(n^2)$-algorithm for finding a weak-path in $G'$ (if exists), i.e., a path in $G'$ from $(a_1,b_1)$ to $(a_n, b_n)$ that consists of forward and possibly also backward edges.    

We describe the weak-path finding algorithm through an analogy to maze traversal.
Notice that $G'$ is a planar graph. We view $G'$ as a maze. Each vertex is a room with four doors, one for each outgoing edge, that lead to the adjacent rooms. For a room $(a_i,b_j)$, the north and south doors lead to rooms $(a_{i+1},b_j)$ and $(a_{i-1},b_j)$, respectively, and the east and west doors lead to $(a_i,b_{j+1})$ and $(a_i,b_{j-1})$, respectively. All doors are closed, but some are locked and some are unlocked. More precisely, 
a door is unlocked if and only if the rooms on both its sides correspond to valid positions.  

\begin{wrapfigure}{r}{0.5\textwidth}
	\vspace{-41pt}
	\begin{center}
		\includegraphics[width=0.4\textwidth]{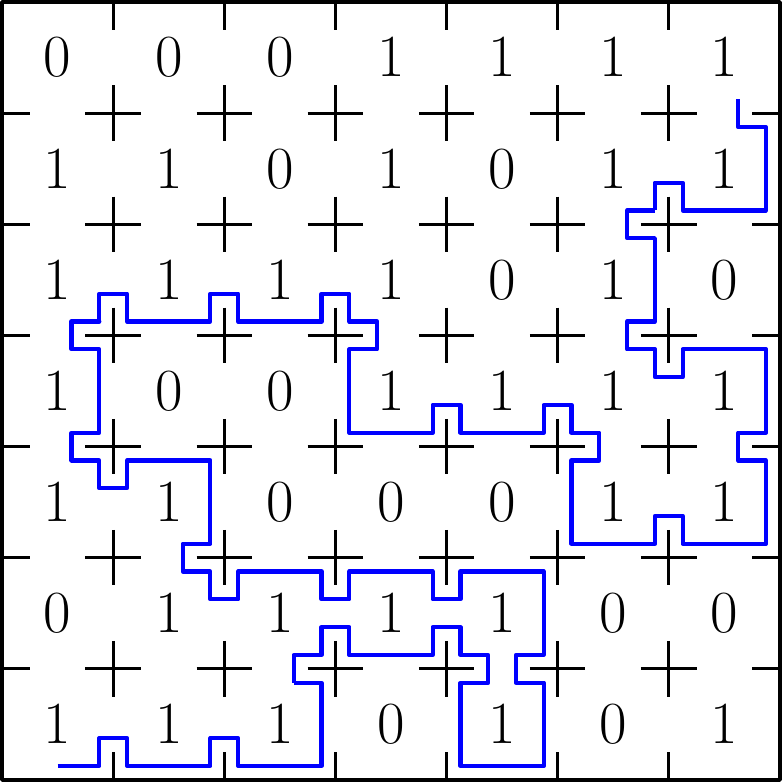}
	\end{center}
	\vspace{-15pt}
	\caption{\small Following the walls (right-hand rule).}
	\vspace{-15pt}
	\label{fig:right-hand}
\end{wrapfigure}
A well known algorithm for traversing a maze is the \emph{wall-follower} rule (also known as the right-hand rule): keep your right hand in contact with a wall of the maze throughout the motion. The algorithm guarantees that you will eventually reach the exit, if possible. (In our setting, whenever our hand encounters an unlocked door, we go through it as if the wall ended at the door.)
The wall-follower rule works only for simply connected mazes, i.e., where all the walls of the maze are connected to the outer boundary of the maze, either directly or indirectly through other walls. In our setting, there might be rooms that are not connected to the outer boundary (when considering graphs for submatrices), but the rule will still work, since the starting and ending points are on the boundary.

Thus, starting from the southern wall of room $(a_1,b_1)$ and using the right-hand rule, one can find in $O(n^2)$ time a weak-path in $G'$ from $(a_1,b_1)$ to $(a_n,b_n)$, if such a path exists (see Figure~\ref{fig:right-hand}).
If no such path exists, the right-hand rule will lead back to $(a_1,b_1)$, and the decision algorithm will return ``no''.

\old{
\begin{algorithm}[H]
	\vspace*{.25cm}
	\begin{enumerate}
		\item If $f(a_{1},b_{1})=0$ or $f(a_{n},b_{n})=0$ return ``no''.
		\item $curr\leftarrow(a_{1},b_{1})$.
		\item $direction\leftarrow 0$ (Directions: $east=0,south=1,west=2,north=3$).
		\item While $curr\ne (a_{n},b_{n})$:
		\begin{description} 
			\item [If] $curr[direction]\ne null$ and $f(curr[direction])=1$
			\begin{description} 
				\item $curr\leftarrow curr[direction]$.
				\item $direction=(direction+1)\%4$ (take a right turn).
			\end{description}
			\item [Else] ($f(curr[direction])=0$)
			\begin{description} 
				\item $direction=(direction-1)\%4$  (take a left turn).
			\end{description}
			\item [If] $curr=(a_{1},b_{1})$, return ``no''.
		\end{description}
		\item Return ``yes''.
	\end{enumerate}
	\caption{dFW-decision($A,B,f$)}\label{alg:FDW}
\end{algorithm}
}

We note that other algorithms exist for deciding whether a weak-path in $G'$ exists (for a given function $f$), but since this algorithm is greedy, it enables us to use the ``fixed distances'' idea that was used in the previous section.

\subsection{The construction of the graphs}

The notion of fixed and non-fixed distances (introduced in Section~\ref{sec:DFGS}) is relevant here as well. Since the rooms in our analogy correspond to vertices in the graph, which in turn correspond to distances in $D$, we simply use the terms: fixed 1-room (or fixed 0-room) for a room which corresponds to a distance that is fixed valid (or fixed non-valid), and non-fixed room for a room that corresponds to a non-fixed distance.

Initially, we have the graph $G'$. Each vertex (or room) $v_{i,j}=(a_i,b_j)$ has four outgoing (directed) edges (or doors):
(i) $north(v_{i,j})=v_{i+1,j}$ (and if $i=n$, $north(v_{i,j})=null$),
(ii) $east(v_{i,j})=v_{i,j+1}$ (and if $j=n$, $east(v_{i,j})=null$),
(iii) $south(v_{i,j})=v_{i-1,j}$ (and if $i=1$, $south(v_{i,j})=null$), and 
(iv) $west(v_{i,j})=v_{i,j-1}$ (and if $j=1$, $west(v_{i,j})=null$).

Now, let $M_{i-1}$ and $G_{i-1}$ be a matrix in level $i-1$ and the graph constructed for it, and let $M_i$ be one of the two submatrices of $M_{i-1}$ in level $i$. We describe how $G_i$ is obtained from $G_{i-1}$. The vertices of $G_i$ correspond to the distances in $M_i$ (i.e., we remove from $V_{i-1}$ the vertices that are fixed w.r.t. $M_i$).

\begin{figure*}[!h]
	\begin{center}
		\includegraphics[width=0.5\textwidth]{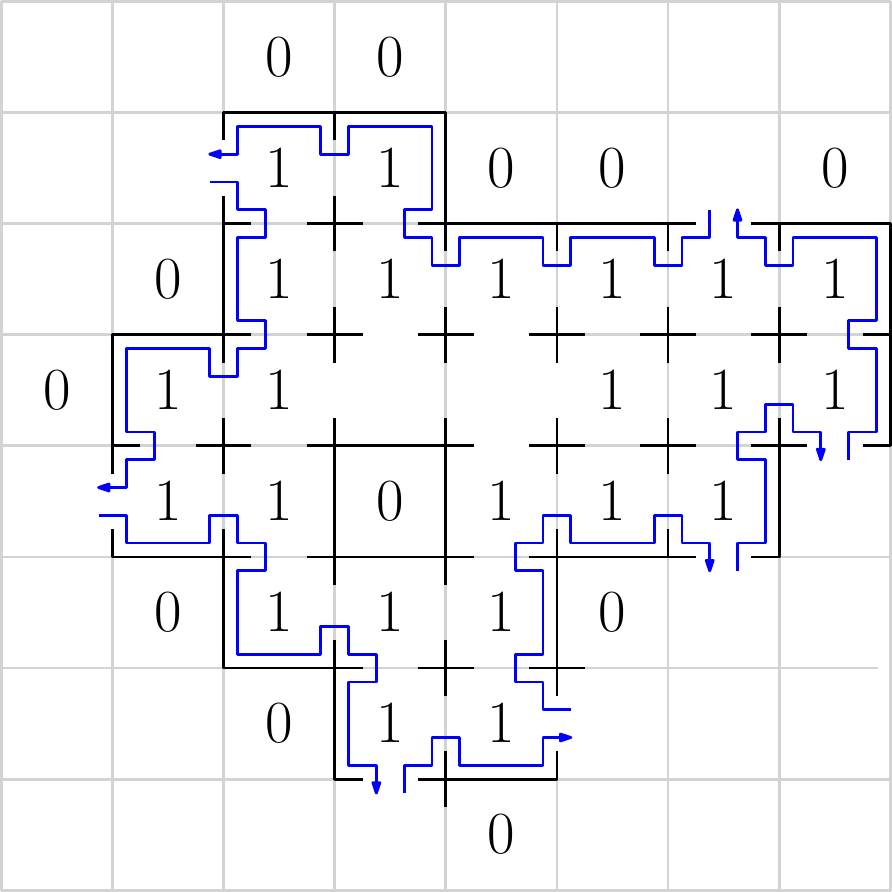}
	\end{center}
	\vspace{-15pt}
	\caption{\small Following the walls (right-hand rule) inside the large room created by a connected component of fixed 1-rooms. The missing doors are due to fixed 0-rooms. All the doors are to non-fixed rooms.}

	\label{fig:1-walls}
\end{figure*}

Returning to the maze analogy, we would like to replace the doors leading to fixed 0-rooms by walls (we already know they are locked), and to remove the doors between pairs of adjacent fixed 1-rooms (we already know they are open). By removing doors we get larger rooms with more than four outgoing doors, but for each way of entering such a large room, there is exactly one way to exit it, using the wall-follower rule. So each incoming edge can be replaced by an edge directly to the next non-fixed room. Notice that if the large room has holes in it, they can be ignored since the maze's exit is in the north-east corner and thus cannot be in the interior of a large room.

$G_i$ is constructed from $G_{i-1}$ as follows:
\begin{enumerate}
	\item Delete all edges $(u,v)$ such that $u$ is a non-fixed room and $v$ is a fixed 0-room.
	\item Let $G_{i-1}'$ be the graph induced by the fixed 1-rooms. For each connected component $C$ of $G_{i-1}'$:
	\begin{itemize}
		\item Let $v_{k,l}$ be any east-most room in $C$. Clearly, this room is adjacent to the outer boundary of $C$ (i.e., to the outer boundary of the union of the rooms in $C$).
		\item Starting from the eastern wall of room $v_{k,l}$, use the wall-follower rule to traverse the outer boundary of $C$. For each visited room $v$, let $next(v)$ be the first non-fixed room that would follow $v$ if we were walking in the graph $G_{i-1}$ (using the wall-follower rule); see Figure~\ref{fig:1-walls}.
	\end{itemize} 
	\item For each edge $(u,v)$ such that $u$ is a non-fixed room and $v$ is a fixed 1-room, replace $(u,v)$ by $(u,next(v))$.
	\item Finally, delete all the fixed rooms. 
\end{enumerate}


\paragraph{Running time.}
Computing $G_{i-1}'$ and its connected components can be done in $O(|G_{i-1}|)$ time using a DFS algorithm. For each connected component $C$, finding its outer boundary can be done in time linear in the size of $C$, using the wall-follower rule. Finding $next(v)$ for each visited fixed 1-room $v$ can be done by walking along the outer boundary in the opposite direction. All the other steps require time linear in the size of $G_{i-1}$. We conclude that $G_i$ can be constructed in $O(|G_{i-1}|)$ time.

\paragraph{Correctness.}
The following lemma is analogous to Lemma~\ref{lem:fixed-paths}.
\begin{lemma}\label{lem:fixed-paths2}
	Given a range $[s,t]$ corresponding to $M_i$, let $\Pi_{i-1}$ be the path traced by the decision algorithm in $G_{i-1}$.
	Let $v_k$ and $v_l$ be two vertices in $\Pi_{i-1}$, such that $v_k$ and $v_l$ are non-fixed w.r.t. $G_i$ but all the vertices between them are fixed w.r.t. $G_i$. Let $u$ be the successor of $v_k$ in $\Pi_{i-1}$. Then,
	\begin{itemize}
		\item if $u=north(v_k)$ in $G_{i-1}$, then $north(v_k)=v_l$ in $G_i$.
		\item if $u=east(v_k)$ in $G_{i-1}$, then $east(v_k)=v_l$ in $G_i$.
		\item if $u=south(v_k)$ in $G_{i-1}$, then $south(v_k)=v_l$ in $G_i$.
		\item if $u=west(v_k)$ in $G_{i-1}$, then $west(v_k)=v_l$ in $G_i$.
	\end{itemize}
\end{lemma}  



\begin{lemma}
	For a given range $[s,t]$, the decision algorithm applied to $G_i$ returns ``yes'' if and only if the decision algorithm applied to $G'$ returns ``yes''.
\end{lemma}
\begin{proof}
	As in Lemma~\ref{lem:shortcuts-correctness}, the proof is by induction on $i$, the level of the recursion.
	Assume that the lemma is true for $G_{i-1}$. 	
	Let $\Pi_{i-1}$ be the path traced by the decision algorithm in $G_{i-1}$.
	Consider the path $\Pi_i$ that is obtained from $\Pi_{i-1}$ by removing all the vertices that are fixed w.r.t. $G_i$.
  We claim that the decision algorithm applied to $G_i$ follows the path $\Pi_i$.
	Indeed, for any two consecutive vertices $v_k$ and $v_{k+1}$ in $\Pi_i$, let $u$ be the successor of $v_k$ in $\Pi_{i-1}$.
	By Lemma~\ref{lem:fixed-paths2}, if $u=north(v_k)$ in $G_{i-1}$ then $north(v_k)=v_{k+1}$ in $G_i$, and the same holds for $east(v_k)$, $south(v_k)$ and $west(v_k)$. So $\Pi_i$ is a path in $G_i$ and moreover, since the decision algorithm makes only local decisions depending on which direction it has currently chosen, it will follow the same path in $G_i$.
	
	Now, let $v_q$ be the last vertex in $\Pi_{i-1}$.
	If the decision algorithm on $G_{i-1}$ returns ``yes'', then $v_q = v_{n,n}$ and the decision algorithm on $G_i$ returns ``yes'' as well.
	If, on the other hand, the decision algorithm on $G_{i-1}$ returns ``no'', then $v_q = v_{1,1}$ and the decision algorithm on $G_i$ will also return ``no''.
	
\old{
	Similarly to Lemma~\ref{lem:shortcuts-correctness}, the proof is by induction on the level of nesting of the graphs/matrices. 
	Let $G_{i-1}$ be the graph that was contracted in level $i-1$, such that $G_i$ is one of the two graphs that were constructed from $G_{i-1}$ by applying the above procedure. 
	By the induction hypothesis, the decision algorithm on $G_{i-1}$ returns ``yes'' if and only if the decision algorithm on $G$ returns ``yes''.
	
	Let $\Pi_{i-1}=(v_1,...,v_k)$ be the path that was found by the decision algorithm on the graph $G_{i-1}$, for the given range $[s,t]$.
	Denote by $v_{l_1},...,v_{l_p}$ the vertices of $\Pi_{i-1}$ that are fixed in $G_i$ ($l_1<...<l_j$). By the construction procedure we have that for any $l_j$ if $north(v_{l_j})=v_{{l_j}+1}$ in $G_{i-1}$ then $north(v_{l_j})=v_{l_{j+1}}$ on $G_i$, and the same holds for $east(v_{l_j})$, $south(v_{l_j})$ and $west(v_{l_j})$. That means $\Pi_{i}=(v_{l_1},...,v_{l_p})$ is a path in $G_i$. Moreover, since the validity of the vertices is determined only by the range $[s,t]$, the decision algorithm on $G_i$ will follow exactly the path $\Pi_{i}$.
	
	If the decision algorithm on $G_{i-1}$ returns ``yes'', then $v_k=v_{n,n}$, and the decision algorithm on $G_i$ returns ``yes'' as well.
	Else, if the decision algorithm on $G_{i-1}$ returns ``no'' then $v_k=v_{1,1}$, and since the decision algorithm on $G_i$ follows the path $\Pi_{i}=(v_{l_1},...,v_{l_p})$, it will also return ``no''. 
}
\end{proof}

\section{Missing proofs}\label{apx:missing_proofs}

\dFSDecision*
\begin{proof}
	
	It is easy to see that if the algorithm returns ``yes'' then $P$ is an s-path (in $G$) from $(a_{1},b_{1})$ to $(a_{n},b_{n})$ and hence $(a_{n},b_{n})$ is an s-reachable position in $G$.
	
	Assume now that $(a_{n},b_{n})$ is an s-reachable position in $G$. (Then, in particular, $f(a_1,b_1)=1$ and $f(a_n,b_n)=1$.) 
	We now prove that if a position $v=(a_{i},b_{j})$ is an s-reachable position in $G$, then there exists a position $v'=(a_{i'},b_{j}) \in P$, $i' \le i$, such that $f(v')=1$. In particular, since $(a_n,b_n)$ is an s-reachable position in $G$, there exists a position $(a_{i'},b_n) \in P$ such that $f(a_{i'},b_n)=1$, and when this position becomes the current position the algorithm returns ``yes''.
	
	We prove this claim by induction on $j$.  	
	The base case where $j=1$ is trivial, since $(a_1,b_1) \in P$.
	Let $\Pi$ be an s-path from $(a_1,b_1)$ to $v=(a_i,b_{j+1})$. Let $u=(a_k,b_j)$, $k \le i$, be a position in $\Pi$ such that $f(u)=1$. $u$ is an s-reachable position in $G$, so by the induction hypothesis there exists a vertex $v'=(a_{i'},b_j) \in P$, $i'\le k$, such that $f(a_{i'},b_j) = 1$. After adding $v'$ to $P$, the algorithm sets $curr\leftarrow(a_{i'},b_{j+1})$. If $f(curr)=1$, then we are done. Else $f(curr)=0$ and the algorithm increases $i'$ until $curr=(a_{i'+l},b_{j+1})$, for some $l \ge 1$, and $f(curr)=1$. Since $f(v)=1$, we conclude that $i'+l \le i$, and the claim follows.
\end{proof}

\totalSize*
\begin{proof}
	By induction on the level. The only matrix in level 0 is $M$, and $|M|=m+k$. Let $M'$ be a matrix in level $i-1$, and assume the size of $M'$ is $p+q$ (it has $p$ rows and $q$ columns). In level $i$ we perform a binary search in the middle row of $M'$ to find the smallest feasible range $[d^{\min}_\frac{p}{2},d^{\max}_j]$ in this row. It is easy to see that the resulting two submatrices are of sizes $\frac{p}{2}+q-j$ and $\frac{p}{2}+j-1$, respectively, which sums to $p+q-1$.
\end{proof}

\subsection{The construction of the graphs in Section \ref{sec:DFGS} - correctness proof}

We assume below that both the distances corresponding to $v_{1,1}$ and to $v_{n,n}$, respectively, are within the range, since otherwise the claim is clearly true.     

\begin{restatable}{lemma}{fixedPaths}
	\label{lem:fixed-paths}
	Given a range $[s,t]$ corresponding to $M_i$, let $\Pi_{i-1}$ be the path traced by the decision algorithm in $G_{i-1}$.
	Let $v_k$ and $v_l$ be two vertices in $\Pi_{i-1}$, such that $v_k$ and $v_l$ are non-fixed w.r.t. $G_i$ but all the vertices between them are fixed w.r.t. $G_i$. Then in $G_i$, $jumpB(v_k)=v_l$, if $f_{[s,t]}(v_k) = 1$, and $jumpA(v_k)=v_l$, if $f_{[s,t]}(v_k) = 0$.
\end{restatable}
\begin{proof}
	First, observe that if $v_l$ immediately follows $v_k$ in $\Pi_{i-1}$, then the lemma is clearly true, so
	let $u_1,\ldots,u_p$ be the vertices between $v_k$ and $v_l$.
	We show by induction that for any $1 \le j \le p$, $next(u_j)=v_l$.
	For $j=p$, if $u_p$ is fixed valid, then by the decision algorithm $v_l \leftarrow jumpB(u_p)$ and by the code fragment $next(u_p)=jumpB(u_p)=v_l$.
	If, on the other hand, $u_p$ is fixed non-valid, then by the decision algorithm $v_l = jumpA(u_p)$ and by the code fragment $next(u_p) \leftarrow jumpA(u_p)=v_l$.
	For $j < p$, if $u_j$ is valid, then by the decision algorithm $u_{j+1}=jumpB(u_j)$ and by the code $next(u_j)\leftarrow next(jumpB(u_j))=next(u_{j+1})$,
	but by the induction hypothesis $next(u_{j+1})=v_l$, so we get $next(u_j)=v_l$.
	If, on the other hand, $u_j$ is non-valid, then by the decision algorithm $u_{j+1}=jumpA(u_j)$ and by the code and the induction hypothesis $next(u_j)\leftarrow next(jumpA(u_j))=next(u_{j+1})=v_l$. 
	
	Now, if $f_{[s,t]}(v_k) = 1$, then by the decision algorithm $u_1=jumpB(v_k)$ and by the code $jumpB(v_k) \leftarrow next(jumpB(v_k))=next(u_1)=v_l$, and,
	if $f_{[s,t]}(v_k) = 0$, then by the decision algorithm $u_1=jumpA(v_k)$ and by the code $jumpA(v_k) \leftarrow next(jumpA(v_k))=next(u_1)=v_l$.
\end{proof}

Notice that Lemma~\ref{lem:fixed-paths} remains true when $v_k=v_{1,1}$, even if $v_{1,1}$ is fixed w.r.t. $G_i$, and when $v_l = v_{n,n}$, even if $v_{n,n}$ is fixed w.r.t. $G_i$.   

\shortcutsCorrectness*
\begin{proof}
	By induction on $i$, the level of the recursion.
	We omit the proof for $G_0$ (i.e., that the decision on $G_0$ is the same as the one on $G$), since it is essentially identical to the proof of the general case.
	We thus assume that the lemma is true for $G_{i-1}$ (i.e., that the decision on $G_{i-1}$ is the same as the one on $G$), and prove that it is also true for $G_i$. 
	
	Let $\Pi_{i-1}$ be the path traced by the decision algorithm in $G_{i-1}$.
	Consider the path $\Pi_i$ that is obtained from $\Pi_{i-1}$ by removing all the vertices that are fixed w.r.t. $G_i$ (except for $v_{1,1,}$ and $v_{n,n}$, even if they are fixed w.r.t. $G_i$). We claim that the decision algorithm applied to $G_i$ follows the path $\Pi_i$.
	Indeed, by Lemma~\ref{lem:fixed-paths}, for any two consecutive vertices $v_k$ and $v_{k+1}$ in $\Pi_i$, if $f_{[s,t]}(v_k) = 1$ then $jumpB(v_k)=v_{k+1}$ and if $f_{[s,t]}(v_k) = 0$ then $jumpA(v_k)=v_{k+1}$, so $\Pi_i$ is a path in $G_i$ and moreover it is followed by the decision algorithm in $G_i$.
	
	Now, let $v_q$ be the last vertex in $\Pi_{i-1}$.
	If the decision algorithm on $G_{i-1}$ returns ``yes'', then $v_q = v_{n,n}$ and the decision algorithm on $G_i$ returns ``yes'' as well.
	If, on the other hand, the decision algorithm on $G_{i-1}$ returns ``no'', then $v_q \ne v_{n,n}$ and it holds that either $f_{[s,t]}(v_q)=1$ and $jumpB(v_q)=null$ or $f_{[s,t]}(v_q)=0$ and $jumpA(v_q)=null$. If $v_q$ is non-fixed w.r.t. $G_i$, then by the code fragment $jumpB(v_q)$ and $jumpA(v_q)$ do not change and the decision algorithm on $G_i$ returns ``no'' in this case.
	If $v_q$ is fixed w.r.t. $G_i$, then by the code fragment $next(v_q)=null$. Let $v_l$ denote the last vertex in $G_{i-1}$ that is non-fixed w.r.t. $G_i$ (i.e., $v_l$ is the last vertex in $\Pi_i$). Arguing as in the proof of Lemma~\ref{lem:fixed-paths}, we get that (in $G_i$) if $f_{[s,t]}(v_l) = 1$ then $jumpB(v_l)=null$ and if $f_{[s,t]}(v_l) = 0$ then $jumpA(v_l)=null$, thus the decision algorithm on $G_i$ returns ``no'' also in this case.  
\end{proof}

\end{document}